%% file: main.tex
\documentclass[submission,copyright,creativecommons]{eptcs}
\input{macros}
\begin{document}
\maketitle

\begin{abstract}
The usual reading of logical implication $\type\to\typetwo$ as `` \emph{if $\type$ then $\typetwo$} '' fails in intuitionistic logic: there are formulas $\type$ and $\typetwo$ such that $\type\to\typetwo$ is not provable, even though $\typetwo$ is provable whenever $\type$ is provable.
Intuitionistic rules apparently don't capture interesting meta-properties of the logic and, from a computational perspective, the programs corresponding to intuitionistic proofs are not powerful enough. 
Such non-provable implications are nevertheless \emph{admissible}, and we study their behaviour by means of a proof term assignment
and related rules of reduction.
We introduce \CalcV{}, a calculus that is able to represent admissible inferences, while remaining in the intuitionistic world by having normal forms that are just intuitionistic terms.
We then extend intuitionistic logic with principles corresponding to admissible rules.
As an example, we consider the Kreisel-Putnam logic \KP{}, for which we prove the strong normalization and the disjunction property through our term assignment.
This is our first step in understanding the essence of admissible rules for intuitionistic logic.
\end{abstract}

\input{introduction}
\input{visser}
\input{harrop}
\input{future}

\bibliographystyle{eptcs}
\bibliography{bibliografia}

\appendix
\input{xxx-norm}
\input{xxx-harrop}
\input{xxx-sn}

\end{document}

%% file: macros.tex
\usepackage[english]{babel}
\usepackage{amsmath, amssymb}
\usepackage{verbatim}
\usepackage{caption}
\usepackage{xcolor}
\usepackage{bussproofs}
\usepackage{stmaryrd}
\usepackage[normalem]{ulem}
\usepackage{amsthm}

\providecommand{\event}{CL\&C 2018} % Name of the event you are submitting to
\makeatletter
\@fpsep\textheight
\makeatother
\title{Admissible Tools in the Kitchen\\ of Intuitionistic Logic\\}
\author{Andrea Condoluci
\institute{Department of Computer Science and Engineering \\ Universit\`a di Bologna \\ Bologna, Italy}
\email{andrea.condoluci@unibo.it}
\and
Matteo Manighetti
\institute{INRIA Saclay \& LIX, \'Ecole Polytechnique \\ Palaiseau, France}
\email{mmanighe@lix.polytechnique.fr}
}
\def\titlerunning{Admissible Tools in Intuitionistic Logic}
\def\authorrunning{A. Condoluci \& M. Manighetti}

\theoremstyle{definition}
\newtheorem{theorem}{Theorem}[section]
\newtheorem{corollary}{Corollary}[section]
\newtheorem{lemma}{Lemma}[section]

\newtheorem{definition}{Definition}[section]
\newtheorem{fact}{Fact}[section]

\newcommand{\IPC}{\textbf{IPC}}
\newcommand{\CPC}{\textbf{CPC}}
\newcommand{\KP}{\textbf{KP}}
\newcommand{\HA}{\textbf{HA}}
\newcommand{\AD}{\textbf{AD}}

\newcommand{\subst}[2]{\{#2/#1\}}

\newcommand{\ih}{\emph{i.h.}}
\newcommand{\ie}{\emph{i.e.}}

\newcommand{\harr}{\mathtt{hop}}
\newcommand{\mmid}{\mathrel{{\mymid}\!\!\!{\mymid}}}
\newcommand{\mymid}{\mathrel{\normalfont\texttt{|}}}
\newcommand{\harrop}[3]{{\normalfont\harr\texttt{[}{#1} \mmid {#2}\mymid{#3}\texttt{]}}}
\newcommand{\Visser}[5]{{\normalfont\mathtt{V}_{#1}\texttt{[}{#2} \mmid {#3}\mymid{#4}\mmid #5\texttt{]}}}
\newcommand{\Pair}[2]{\texttt{<}\,#1\texttt{,}\, #2\texttt{>}}

\newcommand{\EFQ}[1]{\mathtt{efq}\,\,#1}
\newcommand{\Case}[3]{{\normalfont \texttt{case[}#1 \mmid #2 \mymid #3\texttt{]}}}
\newcommand{\In}[2]{\mathtt{inj}_{#1}\,#2}
\newcommand{\Proj}[1]{\mathtt{proj}_{#1}\,}
\newcommand{\Ctx}{W}
\newcommand{\ECtx}{E}
\newcommand{\Bind}[1]{#1\texttt{.}}
\newcommand{\Admissible}[2]{#1 \, \texttt{/} \, #2}

\newcommand{\dashKP}{\vdash_{\KP{}}}
\newcommand{\dashV}{\vdash_{\CalcV{}}}

\newcommand{\Ax}{ax}
\newcommand{\tm}{t}
\newcommand{\CtxHole}{\Box}

\newcommand{\tmtwo}{s}
\newcommand{\tmthree}{u}
\newcommand{\var}{x}
\newcommand{\vartwo}{y}
\newcommand{\varthree}{z}
\newcommand{\defeq}{:=}

\newcommand{\reflemma}[1]{Lemma~\ref{l:#1}}
\newcommand{\refthm}[1]{Theorem~\ref{thm:#1}}
\newcommand{\reffact}[1]{Fact~\ref{fact:#1}}
\newcommand{\reflemmap}[2]{Lemma~\ref{l:#1}(\ref{p:#1-#2})}
\newcommand{\refdef}[1]{Definition~\ref{def:#1}}
\newcommand{\refdefp}[2]{Definition~\ref{def:#1}(\ref{p:#1-#2})}

\newcommand{\sub}{\sigma}
\newcommand{\den}[1]{\llbracket{#1}\rrbracket}
\newcommand{\SN}{\text{SN}}
\newcommand{\type}{A}
\newcommand{\typetwo}{B}
\newcommand{\typethree}{C}
\newcommand{\typefour}{D}

\newcommand{\varset}{\mathcal{V}}

\newcommand{\tatom}{p}

\newcommand{\angled}[1]{\langle{#1}\rangle}
\newcommand{\SNe}{\text{N}\text{e}}

\makeatletter
\newcommand\Copy[2]{% #1 is a key, #2 is the text
  \immediate\write\@auxout{\unexpanded{\global\long\@namedef{mytext@#1}{#2}}}%
  #2%
}

\newcommand\Paste[1]{%
  \ifcsname mytext@#1\endcsname
    \@nameuse{mytext@#1}%
  \else
    ``??''
  \fi
}
\makeatother

\newtheorem{innercustomthm}{}
\newenvironment{customthm}[1]
  {\renewcommand\theinnercustomthm{#1}\innercustomthm}
  {\endinnercustomthm}

\newcommand{\CopyLemma}[3]{%
 \begin{lemma}[\Copy{#2-label}{#1}]\label{l:#2}
  \Copy{#2-body}{#3}
 \end{lemma}%
}
\newcommand{\PasteLemma}[1]{%
 \begin{customthm}{\textbf{\reflemma{#1}}}[\Paste{#1-label}]
   \Paste{#1-body}
 \end{customthm}%
}

\newcommand{\CopyTheorem}[3]{%
 \begin{theorem}[\Copy{#2-label}{#1}]\label{thm:#2}
  \Copy{#2-body}{#3}
 \end{theorem}%
}
\newcommand{\PasteTheorem}[1]{%
 \begin{customthm}{\textbf{\refthm{#1}}}[\Paste{#1-label}]
   \Paste{#1-body}
 \end{customthm}%
}

\newcommand{\CalcV}{\textbf{V}}

\newcommand{\elctx}{K}%{\mathcal{W}}

\newcommand{\toe}{\to_{\SN}}
\newcommand{\toes}{\toe^*}
\newcommand{\toenf}{\twoheadrightarrow_{\SN}}
\newcommand{\toKP}{\to_{\KP{}}}
\newcommand{\toV}{\to_{\CalcV{}}}
\newcommand{\toIPC}{\to_{\IPC{}}}
\newcommand{\mapstoKP}{\mapsto_{\KP{}}}
\newcommand{\mapstoIPC}{\mapsto_{\IPC{}}}
\newcommand{\mapstoV}{\mapsto_{\CalcV{}}}
\newcommand{\ntm}{n}
\newcommand{\Clos}[1]{\overleftarrow{#1}}%{\overline{#1}}

\newcommand{\dom}[1]{\operatorname{dom}(#1)}
\newcommand{\fv}[1]{\operatorname{fv}(#1)}

\newcommand{\eval}[1]{\operatorname{eval}(#1)}
\newcommand{\evali}[1]{\operatorname{eval}_{\IPC}(#1)}

\newcommand{\negneutral}{$\neg$neutral}
\newcommand{\toneutral}{$\to$neutral}
\newcommand{\VisserRule}[1]{Visser$_{#1}$}
\newcommand{\Enf}{$\toe$nf}

%% file: introduction.tex
% !TEX root = admissible.tex
\section{Introduction}
Proof systems are usually presented inductively by giving axioms and rules of inference, which are respectively the ingredients and the tools for cooking new proofs. For example, when presenting \emph{classical propositional logic} (\CPC{}) in \emph{natural deduction}, for each of the usual connectives $\wedge, \vee, \neg, \to, \bot$ one gives a set of standard tools to introduce or remove a connective from a formula in order to obtain a proof.

In their most essential form, we can represent rules as an inference $\Admissible{\type_1,...,\type_n}\typetwo$ (read ``from $\type_1,...,\type_n$ infer $\typetwo$'') where $\type_1,\dots,\type_n,\typetwo$ are schemata of logic formulas. A rule $\Admissible{\type_1,\dots,\type_n}\typetwo$ is said to be \emph{admissible} in a proof system if it is in a way redundant, \ie{} whenever $\type_1\ldots\type_n$ are provable, then $\typetwo$ is already provable without using that rule.
Adding or dropping rules may increase or decrease the amount of proofs we can cook in a proof system. The effect can be dramatic: for example, \emph{classical propositional logic} \CPC{} can be obtained by simply adding the rule of \emph{double negation elimination} ($\Admissible{\neg\neg\type}{\type}$) to \emph{intuitionistic propositional logic} \IPC{}%
%, a system with strikingly different properties\footnote{\Matteo{Qui sarebbe addirittura da citare la disjunction property ma \`e decisamente troppo}}
. Admissible rules are all the opposite: if we decide to utilize one in order to cook something, then we could have just used our ingredients in a different way to reach the same result.

One appealing feature of \CPC{} is the fact that it is \emph{structurally complete}: all its admissible rules are \emph{derivable}, in the sense that whenever $\Admissible{\type_1,\dots,\type_n}{\typetwo}$ is an admissible rule, then also the corresponding principle $\type_1 \land \dots \land \type_n \to \typetwo$ is provable \cite{harrop_disjunctions_1956} --  \ie{} the system acknowledges that there's no need for that additional tool, so we can internalize it and use the old tools to complete our reasoning. This is not the case in intuitionistic logic: the mere fact that we \emph{know} that the tool was not needed, doesn't give us any way to show inside the system \emph{why} is that.
On the other hand, \IPC{} has other wonderful features. Relevant here is the \emph{disjunction property}, fundamental for a constructive system: when a disjunction $\type\lor\typetwo$ is provable, then one of the disjuncts $\type$ or $\typetwo$ is provable as well.
Our interest is in these intuitionistic admissible rules that are not derivable, in the computational principles they describe, and in the logic systems obtained by explicitly adding such rules to \IPC{}.

Can one effectively identify all intuitionistic admissible rules? The question of whether that set of rules is recursively enumerable was posed by Friedman in 1975, and answered positively by Rybakov in 1984. It was then de Jongh and Visser who exhibited a numerable set of rules (now known as \emph{Visser's rules}) and conjectured that it formed a basis for all the admissible rules of \IPC{}. This conjecture was later proved by Iemhoff in the fundamental \cite{iemhoff_admissible_2001}. Rozi\`ere in his Ph.D. thesis \cite{roziere_admissible_1992} reached the same conclusion with a substantially different technique, independently of Visser and Iemhoff.
These works elegantly settled the problem of identifying and building admissible rules. However our question is different: \emph{why} are these rules superfluous, and what reduction steps can eliminate them from proofs?

Rozi\`ere first posed the question of finding a computational correspondence for his basis of the admissible rules in the conclusion of his thesis, but no work has been done on this ever since.
Natural deduction provides a powerful tool to analyse the computational behaviour of logical axioms, thanks to the fact that it gives a simple way to translate axioms into rules and to develop correspondences with $\lambda$-calculi.
Our plan is therefore to understand the phenomenon of admissibility by equipping proofs with $\lambda$-terms and associated reductions in the spirit of the Curry-Howard correspondence.
Normalization will show explicitly what role admissible rules play in a proof.

\subsection{Visser's Basis}
The central role in the developement of the paper is played by Visser's basis of rules. The term \emph{basis} means that any rule that is admissible for \IPC{} is obtainable by combining some of the rules of the family with other intuitionistic reasoning.
It consists of the following sequence of rules:

\[
\text{\VisserRule n}\colon \quad %
\Admissible {%
(\typetwo_i \to \typethree_i)_{i=1\dots n} \to \type_1 \lor \type_2 \,\,}{
\begin{cases}
  \bigvee_{j=1}^n ((\typetwo_i \to \typethree_i)_{i=1\dots n} \to \typetwo_j) \\
  \lor \\
  ((\typetwo_i \to \typethree_i)_{i=1\dots n} \to \type_1) \\
  \lor \\
  ((\typetwo_i \to \typethree_i)_{i=1\dots n} \to \type_2)
\end{cases} %
} \]

This is read as: for every natural number $n$, whenever the left part of the rule (a $n$-ary implication) is provable, then the right part (an $n+2$-ary disjunction) is provable.
It forms a basis in the sense that all other admissible rules of \IPC{} can be obtained from the combination of rules from this family with the usual rules of intuitionistic logic.
It is an infinite family, since \VisserRule {n+1} cannot be derived from \VisserRule 1,\ldots,\VisserRule n~ \cite{roziere_admissible_1992}.

The importance of Visser's basis is not limited to intuitionistic logic but also applies more generally to intermediate logics, as witnessed by the following:

\begin{theorem}[Iemhoff~\cite{iemhoff_intermediate_2005}]
  \label{thm:visser-basis}
  If the rules of Visser's basis are admissible in a logic, then they form a basis for the admissible rules of that logic
\end{theorem}

This theorem also gives us a simple argument to prove the structural completeness of \CPC{}: since all the \VisserRule n rules are provable in \CPC{}, they are admissible and therefore they constitute a basis for \emph{all} the admissible rules of \CPC{}; but since the \VisserRule n are derivable, all admissible rules are derivable.
\subsection{Contributions and Structure of the Paper}
In Section~\ref{sect:terms} we introduce the natural deduction rules corresponding to Visser's rules, and present the associated $\lambda$-calculus $\CalcV{}$:
we show that proofs in the new calculus normalize to ordinary intuitionistic proofs.
In the remaining part of the paper, we push further our idea and start adapting our calculus to intermediate logics characterized by axioms derived from admissible rules.
In Section~\ref{sect:harrop} we study the well-known \emph{Harrop's rule}, and more precisely the logic \KP{}
obtained by adding Harrop's principle to \IPC{}: we prove good properties like subject reduction, the disjunction property, and strong normalization.
In Section~\ref{sect:conclusion} we quickly introduce the logic \AD{} (obtained by adding the axiom $V_1$ to \IPC{}) as a candidate for future study, and possible extensions to arithmetic.
Proofs can be found in the appendices at the end of the paper.
%%% Local Variables:
%%% mode: latex
%%% TeX-master: "admissible"
%%% End:

%% file: visser.tex
% !TEX root = main.tex
\section{Proof Terms for the Admissible Rules: \CalcV{}}\label{sect:terms}

In this section, we are going to assign proof terms to all the inferences of Visser's basis in a uniform way. First, we give
a natural deduction flavor to the Visser rules.
Since the conclusion of the left-hand side of the rules is a disjunction, we model the rules as generalized disjunction eliminations $\vee_E$; ``generalized'' because the main premise will be the disjunction in the antecedent of the \VisserRule n, but under $n$ \emph{implicative} assumptions. Therefore the rules of inference Visser$_n$ have the form:

\begin{prooftree}
  \AxiomC{$\![\typetwo_i \to \typethree_i]_{i=1\dots n}$} \noLine
  \UnaryInfC{\vdots} \noLine \UnaryInfC{$\type_1 \lor \type_2$}

  \AxiomC{$\![(\typetwo_i \to \typethree_i)_{i=1\dots n} \to \type_1]$}
  \noLine \UnaryInfC{\vdots} \noLine \UnaryInfC{$\typefour$}

  \AxiomC{$\![(\typetwo_i \to \typethree_i)_{i=1\dots n} \to \type_2]$}
  \noLine \UnaryInfC{\vdots} \noLine \UnaryInfC{$\typefour$}

  \AxiomC{$\![(\typetwo_i \to \typethree_i)_{i=1\dots n} \to \typetwo_j]_{j=1\dots n}$}
  \noLine \UnaryInfC{$\vdots\quad \quad \quad\quad \vdots$} \noLine \UnaryInfC{$\typefour \quad \cdots \quad \typefour$}
  % \LeftLabel{Visser$_n$}
  \QuaternaryInfC{$\typefour$}
\end{prooftree}

In order to keep the rules admissible, we need to restrict the usage of the inference: the additional requirement is that the proofs of the \emph{main} premise (the one on the left with end-formula $\type_1 \lor \type_2$) must be \emph{closed} proofs, \ie{} cannot have open assumptions others than the ones discharged by that Visser inference. Otherwise we would be able to go beyond \IPC{}, since for example we would prove all the principles corresponding to the admissible rules (as in system \AD{}, see Section~\ref{sect:conclusion}).
On the other side, it is straightforward to see that our rules directly correspond to rules of Visser's basis, and that they adequately represent admissibility.
We now turn to proof terms:

\begin{center}
 \fbox{%
  \begin{minipage}{11cm}\vspace{-0.15in}
    $$\!\!\!\!\begin{array}{lrll}
   \tm, \tmtwo, \tmthree & ::=  & \var, \vartwo, \varthree, \ldots \in \varset \mid \tm\,\tmtwo \mid \lambda\var.\, \tm & \\
   & | & \EFQ \tm & \text{(exfalso)}\\
   & | & \Pair \tm \tmtwo & \text{(pair)}\\
   & | & \Proj i \tm & \text{(projection)}\\
   & | & \In i \tm & \text{(injection)}\\
   & | & \Case \tm {\Bind\vartwo\tmtwo_1} {\Bind\vartwo\tmtwo_2} & \text{(case)}\\
   \\
   & | & \Visser n {\Bind {\vec\var} \tm} {\Bind\vartwo \tmtwo_1} {\Bind\vartwo \tmtwo_2} {\Bind \varthree {\vec \tmthree} } & \text{(Visser --- in \CalcV{} and \AD{})}\\
   & | & \harrop {\Bind {\vec\var} \tm} {\Bind\vartwo \tmtwo_1} {\Bind\vartwo \tmtwo_2} & \text{(Harrop --- in \KP{})}\\
   \end{array}$$
  \end{minipage}}
 \captionof{figure}{Proof terms}
 \label{fig:terms-and-ctxs}
\end{center}

Since the shape of the rules is the elimination of a disjunction, the proof term associated with this inference will be modeled on the \emph{case analysis} $\Case - - -$.
The difference will be in the number of assumptions that are bound, and in the number of possible cases.
We use the vector notation $\Bind {\vec\var} \tm$ on variables to indicate that a sequence of (indexed) variables $\var_1,\ldots,\var_n$ is bound, and on terms like $\Bind\varthree {\vec\tmthree}$ to indicate a sequence of (indexed) terms $\tmthree_1,\ldots,\tmthree_n$ on each of which we are binding the variable $\varthree$.
The resulting annotation for a Visser inference is then:%
     \begin{center}
      \AxiomC{$\vec\var\colon (\typetwo_i \to \typethree_i)_{i=1\dots n} \vdash \tm\colon \type_1 \lor \type_2$}

      \AxiomC{$\Gamma, \vartwo\colon (\typetwo_i \to \typethree_i)_{i=1\dots n} \to \type_1 \vdash \tmtwo_1 : \typefour$}
      \noLine \UnaryInfC{$\Gamma, \vartwo\colon(\typetwo_i \to \typethree_i)_{i=1\dots n} \to \type_2 \vdash \tmtwo_2 : \typefour$}

      \noLine \UnaryInfC{$\{\Gamma, \varthree\colon(\typetwo_i \to \typethree_i)_{i=1\dots n} \to \typetwo_j \vdash \tmthree_j : \typefour\}_{j=1\dots n}$}
      \LeftLabel{Visser$_n$}
      \BinaryInfC{$\Gamma \vdash \Visser n {\Bind {\vec\var} \tm} {\Bind \vartwo {\tmtwo_1}} {\Bind \vartwo {\tmtwo_2}} {\Bind \varthree {\vec \tmthree} } : \typefour$}
      \DisplayProof
    \end{center}
    
    We call \CalcV{} the calculus obtained by adding this family of rules of inference to \IPC{}. The syntax of \CalcV{} can be found in Figure~\ref{fig:terms-and-ctxs}, and it includes the usual proof terms for intuitionistic logic \cite{sorensen_lectures_2006}, plus the proof terms $\Visser - - - - -$ for the Visser family.

We now turn to the reduction rules. First of all, we need to define $\Ctx$ contexts: intuitively, \emph{contexts} are proof terms with a \emph{hole}, where the hole is denoted by $\CtxHole$, and $\ECtx\angled\tm$ means replacing the unique hole in the context $\ECtx$ with the term $\tm$.
% Contexts $\Ctx$ are defined as follows:
\begin{definition}[Weak head \IPC{} contexts]
  $\Ctx$ contexts are defined by the following grammar:
  \[ \Ctx ::= \CtxHole \mid \Ctx\, \tm \mid \Proj i \Ctx
   \mid \Case \Ctx - - .\]
\end{definition}

~\begin{center}
 \fbox{\begin{minipage}{15cm}
 \textbf{Reduction rules for \IPC{}}\vspace{0.5em}\\
 \begin{tabular}{cllcl}
 -- & Beta & $(\lambda\var.\, \tm)\, \tmtwo$ & $\mapsto$ & $\tm\subst\var\tmtwo$ \\
 -- & Projection & $\Proj i {\Pair{\tm_1}{\tm_2}}$ & $\mapsto$ & $\tm_i$ \\
 -- & Case & $\Case{\In{i}\tm}{\Bind\vartwo\tmtwo_1}{\Bind\vartwo\tmtwo_2}$ & $\mapsto$ & $\tmtwo_i\subst\vartwo\tm$ \\
\end{tabular}
\vspace{1em}
\\
 \textbf{Additional rules for \CalcV{}}\vspace{0.5em}\\
 \begin{tabular}{cllcll}
  -- & Visser-inj & $\Visser n {\Bind {\vec\var} {\In{i}\tm}} {\Bind\vartwo\tmtwo_1} {\Bind\vartwo\tmtwo_2} {\Bind\varthree\vec\tmthree}$ & $\mapsto$ & $\tmtwo_i\subst\vartwo{\lambda {\vec\var}.\, \tm}$ & $(i=1,2)$\\
  -- & Visser-efq & $\Visser n {\Bind{\vec\var}\Ctx\angled{\EFQ\tm}} {\Bind\vartwo\tmtwo_1} {\Bind\vartwo\tmtwo_2} {\Bind\varthree\vec\tmthree}$ & $\mapsto$ & $\tmtwo_1\subst\vartwo{(\lambda\vec\var.\,\EFQ\tm)}$ \\
  -- & Visser-app & $\Visser n {\Bind {\vec\var} {\Ctx\angled{\var_j\,\tm}}} {\Bind\vartwo\tmtwo_1} {\Bind\vartwo\tmtwo_2} {\Bind\varthree\vec\tmthree}$ & $\mapsto$ & $\tmthree_j\subst\varthree{\lambda {\vec\var}.\, \tm}$ & $(j=1\ldots n)$\\
 \end{tabular}
\end{minipage}}
\captionof{figure}{Reduction rules (\CalcV{})}
\label{fig:reduction-rules}
\end{center}~

The reduction rules for the proof terms are given in Figure~\ref{fig:reduction-rules}: the first block defines $\mapstoIPC{}$ by means of the usual rules for \IPC{}, and the second block defines $\mapstoV{}$ as $\mapstoIPC{}$ plus additional reduction rules for the new construct $\Visser n {\Bind {\vec\var} \tm'} {\Bind {\vec\vartwo} {\tmtwo_1}} {\Bind\vartwo {\tmtwo_2}} {\Bind\varthree {\vec\tmthree}}$, depending on different shapes that $\tm'$ might have. Let us explain the intuition.
In the first case (\emph{Visser-inj}), the term is the injection $\In{i}\tm$ with possibly free variables $\var_i$ of type $\typetwo_i \to \typethree_i$ for $i=1\ldots n$; in that branch one has chosen to prove one of the two disjuncts $\type_1$ or $\type_2$, and we may just reduce to the corresponding proof $\tmtwo_i$, in which we plug the proof $\tm$ but after binding the free variables $\vec\var$.
In the second case (\emph{Visser-efq}), the disjunction is proved by means of a contradiction, and that contradiction may be used to prove any of the cases $\tmtwo_1,\tmtwo_2,\vec\tmthree$.
In the third case (\emph{Visser-app}), the term contains an application with one of the variables bound by the Visser rule on the left hand side, \ie{} the proof uses one of the Visser assumptions to prove the disjunction. We reduce to the corresponding case $\tmthree_j$, where $\lambda\vec\var.\,\tm$ is substituted for the assumption of type $(\typetwo_i \to \typethree_i)_{i=1\dots n} \to \typetwo_j$.
The reduction relation $\toV{}$ is obtained as usual as the structural closure of the reduction $\mapstoV{}$ (and similarly for $\toIPC{}$).

As expected, \CalcV--terms normalize: we prove normalization
by providing an evaluation function that reduces \CalcV{}--terms
to intuitionistic terms.
The idea is to define the evaluator by structural recursion on typed terms,
and using normalization for \IPC{} after each recursive call.
\PasteTheorem{n-v}

The following is a consequence of \reflemma{eval-wd}:
\begin{theorem}
  \CalcV{}--terms normalize to \IPC{}--terms.
\end{theorem}

%%% Local Variables:
%%% mode: latex
%%% TeX-master: "admissible"
%%% End:

%% file: harrop.tex
% !TEX root = admissible.tex
\section{Beyond \IPC{}: Harrop's Rule and \KP{}}
\label{sect:harrop}
% \section{The Kreisel-Putnam logic \KP{}}\label{sect:KP}
%  The formulas of propositional intuitionistic logic are defined inductively from propositional atoms and logical connectives. We use greek letters like $\alpha,\beta,\gamma,\type,\typetwo$ to denote formulas. The negation is defined as $\neg\type:=\type\to\bot$.

% Harrop's rule \cite{harrop_disjunctions_1956} was the first rule to be shown to be admissible but not derivable in intuitionistic logic. It consist of the following rule of inference
%  \[\Admissible{(\neg\typetwo \to \alpha\lor\beta)\,\,\,}{\,\,\,(\neg\typetwo \to \alpha) \lor (\neg\typetwo \to \beta)}\] % \label{eqn:harrop} \tag{Harrop}\]
% and can be seen as the simplest case arising from Visser's basis of admissible rules. The logic \KP{} is obtained by adding the inference rule above to \IPC{}, and was introduced by G. Kreisel and H. Putnam in  \cite{kreisel_unableitbarkeitsbeweismethode_1957} to exhibit a logic stronger than \IPC{} that preserves the disjunction property, thus disproving the conjecture of \L{}ukasiewicz that no such logic could exist.

In the previous section we have been expecially careful in imposing the restriction on the open assumptions for the application of our new rules, in order to keep our calculus inside the intuitionistic world and to obtain precisely a characterization of admissibility.
At this point, however, one can legitimately ask: what happens if we lift such restriction, and allow one or more admissible principles inside an extended logic?
The system of rules we introduced assumes then a different role, that is the role of providing a simple and modular way to obtain Curry-Howard systems for semi-classical logics arising from the addition to \IPC{} of axioms corresponding to admissible principles.

The simplest and oldest studied admissible rule of \IPC{} is the rule of \emph{independence of premise}, also known as \emph{Harrop's rule} in its propositional variant \cite{harrop_disjunctions_1956}:
\[\Admissible{\neg \typetwo \to \type_1 \lor \type_2\,}{(\neg \typetwo \to \type_1) \lor (\neg \typetwo \to \type_2)} \]
The logic that arises by adding it to \IPC{} has also been studied, and is known as Kreisel-Putnam logic (\KP{}). It was introduced by G. Kreisel and H. Putnam \cite{kreisel_unableitbarkeitbeweismethode_1957} to show a logic stronger than \IPC{} that still could satisfy the disjunction property, thus providing a counterexample to the conjecture of \L{}ukasiewicz that \IPC{} was the only such logic.

We now proceed to define a Curry-Howard calculus for \KP{} as an instance of the system we presented in the previous section. It suffices to realize that Harrop's rule is a particular case of \VisserRule 1 where the formula $\typethree$ is taken to be $\bot$ (note that the third disjunct in this instance of \VisserRule 1 becomes $\neg\typetwo \to \typetwo$, that implies both the other hypotheses $\neg\typetwo\to\type_1$ and $\neg\typetwo\to\type_2$; for this reason we can ignore it). Then we get the following simplified rule in natural deduction:

\begin{center}
 \AxiomC{$[\neg\typetwo]$}\noLine
 \UnaryInfC{$\vdots$}\noLine
 \UnaryInfC{$\type_1\lor\type_2$}
 \AxiomC{$[\neg\typetwo\to\type_1]$}\noLine
 \UnaryInfC{$\vdots$}\noLine
 \UnaryInfC{$\typefour$}
 \AxiomC{$[\neg\typetwo\to\type_2]$}\noLine
 \UnaryInfC{$\vdots$}\noLine
 \UnaryInfC{$\typefour$}
 \LeftLabel{Harrop }
 \TrinaryInfC{$\typefour$}
 \DisplayProof
\end{center}

The restriction on the assumptions of the main premise is now gone, and open proofs are allowed. In fact Harrop's principle is provable in our system:

\begin{center}\centerline{
 \AxiomC{$[\neg \typetwo \to \type_1 \lor \type_2]_{(2)}$}
 \AxiomC{$[\neg\typetwo]_{(1)}$}
 \BinaryInfC{$\type_1\lor\type_2$}
 \AxiomC{$[\neg\typetwo\to\type_1]_{(1)}$}
 \UnaryInfC{$(\neg\typetwo\to\type_1) \lor (\neg\typetwo\to\type_2) $}
 \AxiomC{$[\neg\typetwo\to\type_2]_{(1)}$}
 \UnaryInfC{$(\neg\typetwo\to\type_1) \lor (\neg\typetwo\to\type_2)$}
 \RightLabel{\footnotesize$(1)$}
 \LeftLabel{Harrop}
 \TrinaryInfC{$(\neg\typetwo\to\type_1) \lor (\neg\typetwo\to\type_2)$}
 \RightLabel{\footnotesize$(2)$}
 \UnaryInfC{$(\neg \typetwo \to \type_1 \lor \type_2) \to (\neg\typetwo\to\type_1) \lor (\neg\typetwo\to\type_2)$}
 \DisplayProof
}\end{center}

The proof term is a simplified version of the proof term for $V_1$, where we remove the term corresponding to the trivialized third disjunct:

\begin{center}
 \AxiomC{$\Gamma, \var\colon\neg\typetwo \vdash \tm\colon \type_1\lor\type_2$}
 \AxiomC{$\Gamma, \vartwo\colon\neg\typetwo\to\type_1 \vdash \tmtwo_1\colon \typefour$}
 \AxiomC{$\Gamma, \vartwo\colon\neg\typetwo\to\type_2 \vdash \tmtwo_2\colon \typefour$}
 \TrinaryInfC{$\Gamma \vdash \harrop {\Bind\var\tm} {\Bind\vartwo{\tmtwo_1}} {\Bind\vartwo{\tmtwo_2}} \colon \typefour$}
 \DisplayProof
\end{center}

By inspecting the reduction rules for \CalcV{}, we realize that the rule \emph{Visser-app} has no counterpart in \KP{}: since the Harrop assumptions have negated type, their use in proof terms is completely encapsulated in exfalso terms (see Classification, \reflemma{classification} below). Therefore the reduction rules for \KP{} are the ones for \IPC{} (Figure~\ref{fig:reduction-rules}) plus the additional rules \emph{Harrop-inj} and \emph{Harrop-efq} in Figure~\ref{fig:reduction-rules-kp}.
We denote with $\mapstoKP{}$ the toplevel reduction for \KP{}, and
with $\toKP{}$ its structural closure.

\begin{center}
  \fbox{\begin{minipage}{12.5cm}
 \begin{tabular}{cllcl}
  -- & Harrop-inj & $\harrop {\Bind\var\In{i}\tm} {\Bind\vartwo\tmtwo_1} {\Bind\vartwo\tmtwo_2}$ & $\mapsto$ & $\tmtwo_i\subst\vartwo{\lambda\var.\,\tm}$ \\
  -- & Harrop-efq & $\harrop {\Bind\var\Ctx\angled{\EFQ\tm}} {\Bind\vartwo\tmtwo_1} {\Bind\vartwo\tmtwo_2}$ & $\mapsto$ & $\tmtwo_1\subst\vartwo{(\lambda\var.\,\EFQ\tm)}$ \\
 \end{tabular}
\end{minipage}}
\captionof{figure}{Reduction rules (\KP{})}
\label{fig:reduction-rules-kp}
\end{center}

We prove for \KP{} the usual properties of \emph{subject reduction}, \emph{classification}, and \emph{strong normalization}.
As expected we denote with $\dashKP{}$ the provability in \KP{}, but we use simply $\vdash$ when not ambiguous.
\PasteTheorem{s-r-kp}

In order to classify normal forms of \KP{}, we need to consider
proof terms with possibly open Harrop assumptions: we denote with $\Gamma_\neg$ a \emph{negated} typing context,
\ie{} of the form $\Gamma_\neg = \{ \var_1\colon\neg\type_1, \ldots, \var_n\colon\neg\type_n \}$.
We obtain the following classification of normal forms:
\PasteLemma{classification}

We prove that \KP{} enjoys the strong normalization property, \ie{}
all typable terms are strongly normalizing.
We use a modified version of the method of \emph{reducibility candidates} by Girard-Tait \cite{girard_proofs_1989}.
The differences with respect to the usual proof are that Harrop and exfalso terms are added to neutral terms, and that the reductions for $\harr$
(which involve terms under binders) require special treatment.

\PasteTheorem{kp-sn}

The complete proof is on the appendix. We can now prove the disjunction property:

\begin{lemma}[Consistency]
  $\not\dashKP{} \tm\colon\bot$ for no $\tm$.
\end{lemma}
\begin{proof}
  Let us assume that there exists
  $\tm$ (which we assume in normal form by \refthm{kp-sn}) such that $\dashKP{} \tm\colon\bot$, and derive a contradiction. We proceed by induction on the size of $\tm$.
  The base case is impossible because by \reflemma{classification} $\tm$ cannot be a variable.
  As for the inductive case, by \reflemma{classification}, $\tm$ is either an exfalso, or
  $\var\,\tmthree\in\Gamma_\neg$ for some $\var\in\Gamma_\neg$.
  In the former case $\tm=\EFQ\tmtwo$ for some $\tmtwo$ such that $\dashKP{} \tmtwo\colon\bot$, and we use the \ih{}; the latter case is not possible, since $\Gamma_\neg=\emptyset$.
\end{proof}

\begin{theorem}[Disjunction property]\label{th:disjunctionp}
If $\vdash \type\vee\typetwo$, then $\vdash\type$ or $\vdash\typetwo$.
\end{theorem}
\begin{proof}
  Assume $\vdash \tm\colon \type\vee\typetwo$ for $\tm$ in normal form by \refthm{kp-sn}.
 First note that $\tm\neq\EFQ\tmtwo$, because otherwise by inversion $\vdash \tmtwo\colon\bot$, contradicting consistency.
 By \reflemma{classification} (with $\Gamma_\neg=\emptyset$) $\tm$ is an injection. Conclude by inversion.
\end{proof}

%%% Local Variables:
%%% mode: latex
%%% TeX-master: "admissible"
%%% End:

%% file: future.tex
% !TEX root = admissible.tex
\section{Conclusions and Future Work}\label{sect:conclusion}
Our system provides a meaningful explanation of the admissible rules in terms of normalization of natural deduction proofs. In addition, by simply lifting the condition of having closed proofs on the main premise, we can study intermediate logics characterized by the axioms corresponding to some admissible rules; the study of the Kreisel-Putnam logic exemplifies this approach.

We believe that our presentation is well-suited to continue the study of admissibility in intuitionistic systems, a subject that is currently mostly explored with semantic tools.
We devised powerful proofs of normalization for our systems \KP{} and \CalcV{}, and we will try to extend these results to other similarly obtained systems.
We conclude with some remarks on future generalizations.

\subsection{The Logic \AD{}}
Now that we have shown the potential of our system in analysing the extension of \IPC{} with axioms corresponding to admissible rules, we might wonder what could happen when we try to add several of them. We can be even more ambitious: what if we want to add \emph{all} the Visser rules to \IPC{}?
A theorem by Rozi\`ere greatly simplifies our task:
\begin{theorem}[Rozi\`ere~\cite{roziere_admissible_1992}]
  All Visser rules are derivable in the logic \AD{}, obtained by adding the $V_1$ axiom schema to \IPC{}.
\end{theorem}

Clearly, since the Visser rules are derivable in \AD{} they are also admissible; as we know from Theorem~\ref{thm:visser-basis} this means that they form a basis for all the admissible rules of \AD{}, and since they are derivable we obtain:

\begin{corollary}
  The logic \AD{} is structurally complete.
\end{corollary}

However, we also know from Iemhoff \cite{iemhoff_nother_2001} that \IPC{} is the only logic that has the Visser rules as admissible rules and satisfies the disjunction property. This means that \AD{} cannot satisfy the disjunction property.
This was also proved with different techniques by Rozi\`ere, who also showed that \AD{} is still weaker than \CPC{}.
Given these properties, \AD{} seems the best candidate to be studied with our technique.

\subsection{Arithmetic}

Since its inception with Harrop \cite{harrop_disjunctions_1956}, the motivation for studying admissible rules of \IPC{} was to understand arithmetical systems.
A famous theorem of de Jongh states that the propositional formulas whose arithmetical instances are provable in \emph{intuitionistic arithmetic} (\HA{}) are exactly the theorems of \IPC{}, and many studies of the admissible rules of \HA{} (like Visser \cite{visser_substitutions_2002}, Iemhoff and Artemov \cite{artemov_jonghs_2004}) originated from it.
In particular Visser shows that the propositional admissible rules of \HA{} coincide with those of \IPC{}, and that $\Sigma_1^0$ rules are also related.

Harrop's principle, that we have investigated in this paper, is also known as the propositional Independence of Premise principle.
Its first order version:
\[ (\neg A \to \exists x.\, B(x)) \to \exists x.\, (\neg A \to B(x))  \tag{IP} \]
corresponds to an admissible rule of \HA{} that has an important status in the theory of arithmetic, and was given a constructive interpretation for example by G\"odel~\cite{godel_uber_1958} with his well known \emph{Dialectica} interpretation.

We can assign to IP a proof term and two reduction rules that act in the same way as the ones introduced for Harrop's rule: that is, we will distinguish the two cases where there is an explicit proof of the existential in the antecedent, and where an exfalso reasoning has been carried on.
We believe that a more advanced study of other admissible rules of \HA{} can be carried on similar grounds.
%%% Local Variables:
%%% mode: latex
%%% TeX-master: "admissible"
%%% End:

%% file: xxx-norm.tex
% !TEX root = admissible.tex
\section{Theorems on \CalcV{}}

First some definitions. We denote with $\dashV{}$ the provability in \CalcV{} (but we use $\vdash$ when not ambiguous).
We denote with $\Gamma_\to$ an \emph{implicative} typing context,
\ie{} of the form $\Gamma_\to = \{ \var_1\colon\type_1\to\typetwo_1, \ldots, \var_n\colon\type_n\to\typetwo_n \}$.
We say that a term is \emph{\toneutral} if it has the form
 $\Ctx\angled{\var\,\tmtwo}$ or $\Ctx\angled{\EFQ\tmtwo}$.

\begin{lemma}[Classification for \CalcV{}]\label{l:classification-V}
 Let $\Gamma_{\to} \dashV{} t\colon \type$ for $\tm$ in normal form, and $\tm$ not \toneutral:
\begin{itemize}
 \item \emph{Implication}: if $\type= \typetwo\to\typethree$, then $\tm$ is either an abstraction or a variable in $\Gamma_{\to}$;
 \item \emph{Disjunction}: if $\type= \typetwo\lor\typethree$, then $\tm$ is an injection;
 \item \emph{Conjunction}: if $\type= \typetwo\land\typethree$, then $\tm$ is a pair;
 \end{itemize}
\end{lemma}
\begin{proof}
 By induction on the type derivation of $\tm$:
 \begin{itemize}
  \item ($\Ax$) $\tm$ is a variable in $\Gamma_\to$. By definition of $\Gamma_\to$, the type of $\tm$ is an implication, and we conclude.
  \item ($\to_I$) $\tm$ is an abstraction, and we conclude.
  \item ($\to_E$) and $\tm=\tmtwo\,\tmthree$ with $\Gamma_\to \vdash \tmtwo\colon \typetwo\to\typethree$. Because $\tm$ is in normal form, $\tmtwo$ cannot be an abstraction. By \ih{}, $\tmtwo$ is either a variable in $\Gamma_\to$ or is \toneutral; in both cases $\tm$ is \toneutral{}.
  \item ($\lor_I$) $\tm$ is an injection, and we conclude.
  \item ($\lor_E$) and $\tm=\Case \tmtwo - -$ with $\Gamma_\to\vdash \tmtwo \colon \typefour\lor\typefour'$. Because $\tm$ is in normal form, $\tmtwo$ cannot be an injection.
  By \ih{} $\tmtwo$ is \toneutral, and therefore $\tm$ is \toneutral{}.
  \item ($\land_I$) $\tm$ is a pair, and we conclude.
  \item ($\land_E$) and $\tm=\Proj i \tmtwo$ with $\Gamma_\to\vdash \tmtwo \colon \typefour\land\typefour'$. Because $\tm$ is in normal form, $\tmtwo$ cannot be a pair. By \ih{} $\tmtwo$ is \toneutral, and therefore $\tm$ is \toneutral{}.
  \item (Visser$_n$) not possible. Assume $\tm=\Visser n {\Bind{\vec\var}\tmtwo} - - -$ with $\vec \var \colon (\type_i \to \typetwo_i)_{i=1\dots n} \vdash \tmtwo\colon \type_1\lor\type_2$ by inversion, and derive a contradiction. By \ih{} $\tmtwo$ is \toneutral{} or an injection, but both cases contradict the hypothesis that $\tm$ is a normal form.
 \end{itemize}
\end{proof}

In order to prove normalization, we define an evaluation function $\eval\cdot$,
 mapping each typable term in \CalcV{} to its normal form.
 We first assume a corresponding function for \IPC{}:

\begin{definition}[$\evali\cdot$]
  We call $\evali\cdot$ the function mapping each term typable in \IPC{} to its normal form.
\end{definition}

\begin{definition}[$\eval\cdot$]
   Let $\tm$ a term typable in \CalcV. 
   We define its evaluation $\eval\tm$ by structural induction:
\[\begin{array}{ll} 
  \eval{\var} & \defeq \var \\
  \eval{\tm\,\tmtwo} & \defeq \evali{\eval\tm\,\eval\tmtwo} \\
  \eval{\lambda\var.\,\tm} & \defeq \lambda\var.\,\eval\tm \\
  \eval{\EFQ\tm} & \defeq \EFQ{(\eval\tm)} \\
  \eval{\angled{\tm,\tmtwo}} & \defeq \angled{\eval\tm,\eval\tmtwo} \\
  \eval{\Proj i \tm} & \defeq \evali{\Proj i {(\eval \tm)}} \\
  \eval{\In i \tm} & \defeq \evali{\In i {(\eval \tm)}} \\
  \eval{\Case \tm {\Bind\vartwo\tmtwo_1} {\Bind\vartwo\tmtwo_2}} & \defeq \evali{\Case {\eval\tm} {\Bind\vartwo\eval{\tmtwo_1}} {\Bind\vartwo\eval{\tmtwo_2}}} \\
  \eval{\Visser n {\Bind {\vec\var} \tm} {\Bind\vartwo \tmtwo_1} {\Bind\vartwo \tmtwo_2} {\Bind \varthree {\vec \tmthree} }} & \defeq \begin{cases}
    \evali{\eval{\tmtwo_i}\subst\vartwo{\lambda {\vec\var}.\, \tm'}} & \text{if } \eval\tm=\In{i}{\tm'} \\
    \evali{\eval{\tmtwo_1}\subst\vartwo{\lambda {\vec\var}.\, \EFQ{\tm'}}} & \text{if } \eval\tm=\Ctx\angled{\EFQ{\tm'}} \\
    \evali{\eval{\tmthree_j}\subst\varthree{\lambda {\vec\var}.\, \tm'}} & \text{if } \eval\tm=\Ctx\angled{\var_j\,\tm'} \\
  \end{cases} \\
\end{array}\]
Note: the three cases in the definition of $\eval\cdot$ on Visser terms
 are exhaustive by inspection of the normal forms of type disjunction (\reflemma{classification-V})
 since it holds by inversion that $\Gamma_\to\vdash\tm\colon\type_1\lor\type_2$ with $\dom{\Gamma_\to}=\vec\var$.
\end{definition}

\begin{lemma}[$\eval\cdot$ well-defined]\label{l:eval-wd} For every \CalcV-term $\tm$ s.t. $\Gamma\vdash_{\text{\CalcV}} \tm\colon\type$:
  \begin{enumerate}
    \item $\Gamma \vdash_{\IPC{}} \eval\tm \colon \type$,
    \item $\eval\tm$ is normal,
    \item $\tm\to_{\CalcV{}}^*\eval\tm$.
  \end{enumerate}
\end{lemma}
\begin{proof}
  The three points can be proved mutually, by 
  induction on the type derivation $\Gamma\vdash_{\text{\CalcV}} \tm\colon\type$:
  \begin{enumerate}
    \item follows by \ih{} and by subject reduction for \IPC{};
    \item follows by \ih{} and from the fact that the output of $\evali\cdot$ are only normal forms;
    \item follows by \ih{} and from the fact that \IPC{} is a subcalculus of \CalcV{}.
  \end{enumerate}
\end{proof}

It easily follows:
\CopyTheorem{Normalization for \CalcV{}}{n-v}{%
  \CalcV{} enjoys the normalization property.%
}

%%% Local Variables:
%%% mode: latex
%%% TeX-master: "admissible"
%%% End:

%% file: xxx-harrop.tex
% !TEX root = main.tex
\section{Theorems on \KP{}}%
\CopyTheorem{Subject reduction for \KP{}}{s-r-kp}{
  If $\Gamma \dashKP{} \tm\colon\type$ and  $\tm \toKP{} \tmtwo$, then $\Gamma \dashKP{} \tmtwo \colon \type$.
}
\begin{proof}
 By the definition of reduction as the closure of $\mapstoKP$ under evaluation contexts, we just prove the statement when $\tm\mapstoKP \tmtwo$; the general case $\tm\toKP\tmtwo$ follows because substitution preserves types.

 The cases of the usual intuitionistic reductions are standard (see for example \cite{sorensen_lectures_2006}); we just prove the cases of the reduction rules associated with $\harr$.

 For the case of the left injection $\harrop {\Bind\var\In 1 \tm} {\Bind \vartwo \tmtwo_1} {\Bind\vartwo\tmtwo_2} \mapsto \tmtwo_1\subst\vartwo{\lambda\var.\,\tm}$, by inversion we have $\Gamma, \vartwo\colon \neg\typetwo \to \type_1 \vdash \tmtwo_1 \colon \typefour$ and  $\Gamma, \var\colon \neg\typetwo \vdash  \In 1 \tm \colon \type_1 \lor \type_2$ for some $\type_1,\type_2, \typetwo, \typefour$.
 Again by inversion \\$\Gamma, \var\colon \neg\typetwo \vdash \var\colon \type_1$, and by $\to_I$ we obtain $\Gamma \vdash \lambda\var.\,\tm \colon \neg\typetwo \to \type_1$.
 By substitutivity we get the desired result $\Gamma \vdash \tmtwo_1\subst\vartwo{\lambda\var.\,\tm} \colon \typefour$.
The case of the right injection is analogous.

Finally, if $\harrop {\Bind\var \Ctx\angled{\EFQ\tm}} {\Bind\vartwo\tmtwo_1} {\Bind\vartwo\tmtwo_2} \mapsto \tmtwo_1\subst\vartwo{\lambda\var.\,\EFQ\tm}$, by inversion we have $\Gamma, \vartwo\colon \neg\typetwo \to \type_1 \vdash \tmtwo_1 \colon \typefour$ and $\Gamma, \var\colon \neg\typetwo \vdash  \Ctx\angled{\EFQ\tm} \colon \type_1\lor\type_2$ for some $\type_1, \type_2, \typetwo, \typefour$. It is easy to see, by induction on the definition of weak head contexts and by inversion, that $\Gamma, \var\colon \neg\typetwo \vdash \tm \colon \bot$; by $\bot_E$ we obtain $\Gamma, \var\colon \neg\typetwo \vdash \EFQ\tm \colon \type_1$. By $\to_I$ we obtain $\Gamma \vdash \lambda\var.\, \EFQ\tm \colon \neg\typetwo \to \type_1$, and by substitutivity we get the desired result $\Gamma \vdash \tmtwo_1\subst\vartwo{\lambda\var.\, \EFQ\tm} \colon \typefour$.
\end{proof}

We say that a term is \emph{\negneutral} if it has the form $\Ctx\angled{\EFQ\tm}$.
\CopyLemma{Classification for \KP{}}{classification}{
Let $\Gamma_\neg \dashKP{} t\colon \type$ for $\tm$ in (weak head) normal form and $\tm$ not \negneutral:
\begin{itemize}
 \item \emph{Implication}: if $\type= \typetwo\to\typethree$, then $\tm$ is an abstraction or a variable in $\Gamma_\neg$;
 \item \emph{Disjunction}: if $\type= \typetwo\lor\typethree$, then $\tm$ is an injection;
 \item \emph{Conjunction}: if $\type= \typetwo\land\typethree$, then $\tm$ is a pair;
 \item \emph{Falsity}:  if $\type= \bot$, then $\tm = \var\,\tmtwo$ for some $\tmtwo$ and some $\var\in\Gamma_\neg$.
\end{itemize}
}
\begin{proof}
 By induction on the type derivation of $\tm$:
 \begin{itemize}
  \item ($\Ax$) and $\tm$ is a variable in $\Gamma_\neg$: by definition of $\Gamma_\neg$, the type of $\tm$ is an implication, and we conclude.
  \item ($\to_I$) and $\tm$ is an abstraction: trivial.
  \item ($\to_E$) and $\tm=\tmtwo\,\tmthree$ with $\Gamma_\neg \vdash \tmtwo\colon \typetwo\to\type$. Because $\tm$ is in normal form, $\tmtwo$ cannot be an abstraction. By \ih{}, $\tmtwo$ is either a variable in $\Gamma_\neg$ or a \negneutral{} term. In the first case, note that we have that $\type=\bot$ and $\tm = x \tmtwo$, and the thesis holds; in the second case, $t$ is \negneutral{} and the thesis holds.
  \item ($\lor_I$) and $\tm$ is an injection: trivial.
  \item ($\lor_E$) and $\tm=\Case \tmtwo - -$ with $\Gamma_\neg\vdash \tmtwo \colon \typefour_1\lor\typefour_2$.
  By \ih{} $\tmtwo$ is either an injection or \negneutral{}.
  The first case is not possible because $\tm$ is in normal form; in the second case, $\tm$ is \negneutral{} as required.
  \item ($\land_I$) and $\tm$ is a pair: trivial.
  \item ($\land_E$) and $\tm=\Proj i \tmtwo$ with $\Gamma_\neg\vdash \tmtwo \colon \typefour_1\land\typefour_2$. By \ih{} $\tmtwo$ is either a pair or \negneutral{}, but the first case contradicts the hypothesis that $\tm$ is in normal form. Therefore $\tmtwo$ is \negneutral{}, and also $\tm$ is \negneutral{}.
  \item ($\bot_E$) then $\tm$ is immediately \negneutral{}.
  \item (Harrop) not possible. Assume $\tm=\harrop {\Bind\var\tmtwo} - -$ with $\Gamma_\neg, \var\colon\neg\typetwo\vdash \tmtwo\colon \typefour_1\lor\typefour_2$, and derive a contradiction. By \ih{} $\tmtwo$ is an injection or a \negneutral{} term, but both cases contradict the hypothesis that $\tm$ is in normal form.
 \end{itemize}
\end{proof}

%%% Local Variables:
%%% mode: latex
%%% TeX-master: "admissible"
%%% End:

%% file: xxx-sn.tex
% !TEX root = main.tex
\subsection{Strong Normalization}
In this section, we prove the strong normalization property for \KP{}
by means of an adapted version of Girard's method of candidates \cite{girard_proofs_1989}.

\begin{definition}[Weak head \KP{} contexts]
  % We define the set of \emph{weak head contexts},
  % described by entry $\elctx$ in the grammar
  \[ \elctx ::= \Box \mid \elctx\,\tmtwo \mid \Proj i \elctx \mid %\Case\elctx{\tmtwo_1}{\tmtwo_2} \mid
  \Case \elctx {\Bind\vartwo{\tmtwo_1}} {\Bind\vartwo{\tmtwo_2}} \mid %
  \harrop {\Bind\var\elctx} {\Bind\vartwo\tmtwo_1} {\Bind\vartwo\tmtwo_2} \]
\end{definition}

Let $\SN$ be the set of \emph{strongly normalizing terms} of \KP{}.
By abuse of notation, we say that a context -- be it an \IPC{} context $\Ctx$ or a \KP{} context $\elctx$ -- is strongly normalizing if all its ``internal'' $\lambda$-terms are strongly normalizing.

\begin{definition}[Weak head reduction $\toe$, $\toenf$]\label{def:toenf}
  We define $\toe$ as the ``strongly normalizing'' closure of $\mapstoKP{}$ (Figure~\ref{fig:reduction-rules}) under weak head contexts:
    \[\begin{array}{lcl}
    \elctx\angled{(\lambda\var.\, \tm)\, \tmtwo} & \toe & \elctx\angled{\tm\subst\var\tmtwo} 
    \\
    \elctx\angled{\Proj i {\Pair{\tmtwo_1}{\tmtwo_2}}} & \toe & \elctx\angled{\tmtwo_i} \\ \elctx\angled{\Case{\In{i}\tm}{\Bind\vartwo\tmtwo_1}{\Bind\vartwo\tmtwo_2}} & \toe & \elctx\angled{\tmtwo_i\subst\vartwo\tm} \\
    \elctx\angled{\harrop {\Bind\var\In{i}\tm} {\Bind\vartwo\tmtwo_1} {\Bind\vartwo\tmtwo_2}} & \toe & \elctx\angled{\tmtwo_i\subst\vartwo{\lambda\var.\,\tm}} \\
    \elctx\angled{\harrop {\Bind\var\Ctx\angled{\EFQ\tm}} {\Bind\vartwo\tmtwo_1} {\Bind\vartwo\tmtwo_2}} & \toe & \elctx\angled{\tmtwo_1\subst\vartwo{(\lambda\var.\,\EFQ\tm)}} \\
    \end{array}\]
  for every \SN{} contexts $\Ctx,\elctx$ and $\tm,\tmtwo,\tmtwo_1,\tmtwo_2\in\SN$.
  As usual, we denote by $\toe^*$ the reflexive and transitive closure of $\toe$.
  A term $\tm$ is a $\toe$-normal form (in short, \Enf) if $\tm\not\toe$.
  We say that $\tm\toenf\tmtwo$ if $\tm\toe^*\tmtwo$ and $\tmtwo$ is a \Enf.
\end{definition}

By inspection of the reduction rules, one may prove:
\begin{lemma}
  $\toe$ is deterministic.
\end{lemma}

One of the main properties of reducibility candidates is that
they are \emph{backward closed} under reduction:
\begin{definition}[Backward closure $\Clos\cdot$]\label{def:closure}
  Let $T$ be a set of \Enf{s}.
 We define its \emph{closure under backward weak head reduction} as the set
 $\Clos{T} \defeq \{ \tmtwo \mid \tmtwo \toenf \tm \in T \}$.
\end{definition}

\begin{lemma}[Backward closure of $\SN$]\label{l:prec}
  $\SN$ is \emph{backward closed} under $\toe$.
\end{lemma}
\begin{proof}
  Let $\tm\in\SN$ and
  $\tmtwo\toe\tm$; we need show that $\tmtwo\in\SN$.
  By cases on the reduction rules of \refdef{toenf}; we only consider the case of \emph{Harrop-inj}, as one can proceed in a similar way for the other reduction rules.
  Let $ \tmtwo = \elctx\angled{\harrop {\Bind\var\In{i}{\tm'}} {\Bind\vartwo\tmtwo_1'} {\Bind\vartwo\tmtwo_2'}} \toe \elctx\angled{\tmtwo_i'\subst\vartwo{\lambda\var.\,\tm'}} = \tm $, and let us consider a reduction sequence beginning with $\tmtwo$.
  Either the sequence terminates after some internal redutions
  \[ \tmtwo \to^* \elctx'\angled{\harrop {\Bind\var\In{i}{\tm''}} {\Bind\vartwo\tmtwo_1''} {\Bind\vartwo\tmtwo_2''}} \]
  which must terminate because all internal terms are $\SN$ by definition of $\toe$, or eventually we have 
  \[%
  \elctx'\angled{\harrop {\Bind\var\In{i}{\tm''}} {\Bind\vartwo\tmtwo_1''} {\Bind\vartwo\tmtwo_2''}}
  \to
  \elctx'\angled{\tmtwo_i''\subst\vartwo{\lambda\var.\,\tm''}}.
  \]
  This term is strongly normalizing because it is a reduct of $\tm$, and by hypothesis $\tm\in\SN$. Therefore the reduction sequence must terminate.
\end{proof}

Another key notion are \emph{neutral terms},
that are intuitively \Enf{s} that do not begin with constructors:
\begin{definition}[Neutral terms]
  $ \SNe \defeq \{\elctx\angled\var \mid \elctx \text{ is } \SN  \text{ and } \var \text{ a variable}\} \cup \{\Ctx\angled{\EFQ\tm} \mid \Ctx\text{ and }\tm \text{ are } \SN\} $.
\end{definition}
\begin{fact}\label{fact:NeSN}
  Neutral terms are strongly normalizing \Enf{s}.
\end{fact}

We are now ready to define the semantics of formulas:
\begin{definition}[Denotation $\den\cdot$]\label{def:valuation}~
  % We define the semantics $\den\cdot$ of types as follows:
  \begin{enumerate}
    \item\label{p:valuation-atom} $\den{\tatom} \defeq
     \SN$ for every $\tatom$ atomic (also $\tatom=\bot$),
    \item\label{p:valuation-to}
     $\den{\type\to\typetwo} \defeq %
      \Clos{\{\lambda\var.\,\tm \mid \forall \tmtwo\in\den\type,\, \tm\subst\var\tmtwo\in\den\typetwo \}}
      \cup \Clos\SNe $,
    \item\label{p:valuation-and} 
     $\den{\type_1\land\type_2} \defeq %
     \Clos{\{ \Pair {\tm_1} {\tm_2} \mid \tm_i\in\den{\type_i} \}} %
     \cup \Clos\SNe $,
    \item\label{p:valuation-or} 
     $\den{\type_1\lor\type_2} \defeq %
     \Clos{\{\In i \tm \mid \tm \in \den{\type_i} \}} \cup \Clos\SNe $.
  \end{enumerate}
\end{definition}
% \begin{lemma}
%   The denotation of propositions is well-defined.
% \end{lemma}
% \begin{proof}
%   We prove by induction on the structure of propositions that
%   all sets occuring in the definition which are closed with $\Clos\cdot$ contain only strongly normalizing terms in \Enf.
%   \begin{itemize}
%     \item For the set $\SNe$ it follows from \reffact{NeSN}.
%     \item Abstractions are in \Enf. To show that their bodies
%     are strongly normalizing, use the \ih{} and the fact that variables are neutral terms.
%     \item Pairs are in \Enf, and their components are strongly normalizing by \ih{}
%     \item Injections are in \Enf, and their components are strongly normalizing by \ih{}
%   \end{itemize}
% \end{proof}

In fact, we note that our definition produces candidates of reducibility:
\begin{lemma}[Denotations are candidates]\label{l:candidates}
  For every $\type$, its denotation:
  \begin{enumerate}
    \item\label{p:candidates-sn} contains only \emph{strongly normalizing} terms: $\den\type\subseteq\SN$
    \item\label{p:candidates-vars} contains all \emph{neutral terms}: $\SNe \subseteq \den\type$
    \item\label{p:candidates-closure} is \emph{backward closed}: if $\tm\in \den\type$ and $\tmtwo\toe\tm$, then $\tmtwo \in \den\type$.
  \end{enumerate}
\end{lemma}
\begin{proof}
  Points \ref{p:candidates-vars} and \ref{p:candidates-closure} are trivial.
  Before proving Point~\ref{p:candidates-sn} we note that as shown in the proof of \reflemma{prec}, if $T$ contains only strongly normalizing terms, then $\Clos T$ does too. 
  We can then prove Point~\ref{p:candidates-sn} by induction on the structure of types:
  the case of propositional atoms follows from \refdefp{valuation}{atom} and \reflemma{prec};
  for the inductive cases, use \reffact{NeSN}, the \ih{} and \reflemma{prec}.
\end{proof}
% \begin{proof}
%   The points \ref{p:candidates-vars} and \ref{p:candidates-closure} are evident.
%   We show point \ref{p:candidates-sn} by induction on the structure of formulas.
%   In all cases, we use \reffact{NeSN}, \reflemma{SNclosed}, and the \ih{}.
%   It suffices to note that:
%    $\lambda\var.\,\tm$ is $\SN$ iff $\tm$ is $\SN$;
%    $\angled{\tm,\tmtwo}$ is $\SN$ iff $\tm$ and $\tmtwo$ are $\SN$;
%    $\In i \tm$ is $\SN$ iff $\tm$ is $\SN$.
% \end{proof}

We extend the definition of valuation to typing contexts:
\begin{definition}
  Let $\Gamma$ be a typing context;
  we define $\den\Gamma$ as the set of substitutions mapping variables in $\Gamma$ to terms in the denotation of the corresponding type, \ie{}
  \[ \den\Gamma \defeq \{ \sigma \text{ substitution} \mid %
  \dom\sigma = \dom\Gamma \text{ and }
  (\var\mapsto\tm)\in\sigma \text{ implies }%
  \tm\in\den{\Gamma(\var)}\} \]
  where $\Gamma(\var)\defeq \type$ when $(\var\colon\type)\in\Gamma$. 
\end{definition}

A lemma useful in the proof of \reflemma{aux-sn}:
\begin{lemma}\label{l:toes-subst}
  If $\tm\toe\tmtwo$ and $\tm\sigma\in\SN$, then
  $\tm\sigma\toe\tmtwo\sigma$.
\end{lemma}
\begin{proof}
  First note that if $\tm=\elctx\angled{\tm'}$ and $\tm\sigma\in\SN$, then
  $\tm\sigma = \elctx'\angled{\tm'\sigma}$ for some $\SN$ context $\elctx'$.
 Therefore, we assume that $\elctx\angled{\tm'} \toe \elctx\angled{\tmtwo'}$
 with $\tm'\mapsto\tmtwo'$, and we prove that $\tm'\sigma\mapsto\tmtwo'\sigma$
  by cases on the reduction rules:
  \begin{itemize}
    \item $(\lambda\vartwo.\,\tmtwo)\,\tmthree \mapsto \tmtwo\subst\vartwo\tmthree$.
    By renaming, $\vartwo\not\in\fv\sigma,\dom\sigma$.
    Then $((\lambda\vartwo.\,\tmtwo)\,\tmthree)\sigma =
     (\lambda\vartwo.\,\tmtwo\sigma)\,(\tmthree\sigma) \mapsto 
     \tmtwo\sigma\subst\vartwo{\tmthree\sigma}$,
     with $\tmtwo\subst\vartwo\tmthree\sigma =
     \tmtwo\subst\vartwo\tmthree\sigma$.
     We conclude because
     $\tmtwo\sigma\subst\vartwo{\tmthree\sigma}=\tmtwo\subst\vartwo\tmthree\sigma$ and $\tmtwo\sigma,\tmthree\sigma\in\SN$ by the hypothesis that $\tm\sigma\in\SN$.
     \item $\Proj i {\Pair{\tm_1}{\tm_2}} \mapsto \tm_i$.
     Then $(\Proj i {\Pair{\tm_1}{\tm_2}})\sigma = \Proj i {\Pair{\tm_1\sigma}{\tm_2\sigma}}
     \mapsto \tm_i\sigma$, and $\tm_1\sigma,\tm_2\sigma\in\SN$ by hypothesis.
     \item $\Case {\In i \tm} {\Bind\vartwo{\tmtwo_1}} {\Bind\vartwo{\tmtwo_2}} \mapsto \tmtwo_i\subst\vartwo\tm$.
     By renaming, $\vartwo\not\in\fv\sigma,\dom\sigma$. Similar to the case below.
     % \item $\Ctx\angled{\EFQ\tm} \mapsto \EFQ\tm$. Then
     % $(\Ctx\angled{\EFQ\tm})\sigma = \Ctx'\angled{\EFQ{(\tm\sigma)}}$ for some $\Ctx'$.
     % We conclude because $\Ctx'\angled{\EFQ{(\tm\sigma)}} \mapsto \EFQ{(\tm\sigma)} = (\EFQ\tm)\sigma$.
     \item $\harrop {\Bind\var \In i \tm} {\Bind\vartwo{\tmtwo_1}} {\Bind\vartwo{\tmtwo_2}} \mapsto \tmtwo_i\subst\vartwo{\lambda\var.\,\tm}$.
     By renaming, $\var,\vartwo\not\in\fv\sigma,\dom\sigma$.
     \\Then $\harrop {\Bind\var\In i \tm} {\Bind\vartwo{\tmtwo_1}} {\Bind\vartwo{\tmtwo_2}}\sigma = %
     \harrop {\Bind\var\In i (\tm\sigma)} {\Bind\vartwo{\tmtwo_1\sigma}} {\Bind\vartwo{\tmtwo_2\sigma}} 
     \mapsto \tmtwo_i\sigma\subst\vartwo{\lambda\var.\,\tm\sigma} $.
     We conclude because $\tmtwo_i\sigma\subst\vartwo{\lambda\var.\,\tm\sigma} = (\tmtwo_i\subst\vartwo{\lambda\var.\,\tm})\sigma$.
     \item $\harrop {\Bind\var\Ctx\angled{\EFQ\tm}} {\Bind\vartwo{\tmtwo_1}} {\Bind\vartwo{\tmtwo_2}} \mapsto \tmtwo_1\subst\vartwo{\lambda\var.\,\EFQ\tm}$.
      By renaming, $\var,\vartwo\not\in\fv\sigma,\dom\sigma$.
      \\Then $\harrop {\Bind\var\Ctx\angled{\EFQ\tm}} {\Bind\vartwo{\tmtwo_1}} {\Bind\vartwo{\tmtwo_2}}\sigma = \harrop {\Bind\var\Ctx'\angled{\EFQ\tm\sigma}} {\Bind\vartwo{\tmtwo_1\sigma}} {\Bind\vartwo{\tmtwo_2\sigma}}$ for some $\SN$ context $\Ctx'$.
      We have $\harrop {\Bind\var\Ctx'\angled{\EFQ\tm\sigma}} {\Bind\vartwo{\tmtwo_1\sigma}} {\Bind\vartwo{\tmtwo_2\sigma}} %
      \mapsto \tmtwo_1\sigma\subst\vartwo{\lambda\var.\,\EFQ{(\tm\sigma)}}$, and we conclude because $\tmtwo_1\sigma\subst\vartwo{\lambda\var.\,\EFQ{(\tm\sigma)}} = \tmtwo_1\subst\vartwo{\lambda\var.\,\EFQ \tm} \sigma$.
  \end{itemize}
\end{proof}

\begin{lemma}[Fundamental lemma]\label{l:aux-sn}
  If $\Gamma\vdash\tm\colon\type$ and $\sub\in\den\Gamma$,
   then $\tm\sub\in\den\type$.
\end{lemma}
\begin{proof}
  By induction on the type derivation.
  The base case is the axiom, and instantiated variables belong to the corresponding denotations
  by the definition of $\den\Gamma$.
  Let us now proceed by cases on the rules of inference:
  \begin{itemize}
    \item[($\to_I$)]
     Assume that for all $\sub\in\den{\Gamma,\var\colon\type}$, $\tm\sub\in\den\typetwo$;
     we need to prove that for all $\sub\in\den\Gamma$, $(\lambda\var.\,\tm)\sub\in\den{\type\to\typetwo}$.
     Let $\sub\in\den\Gamma$, and 
     by renaming $\var\not\in\dom\sub\cup\fv\sub$. Then $(\lambda\var.\,\tm)\sub = \lambda\var.\,\tm\sub$.
     By \refdefp{valuation}{to}, $\lambda\var.\,\tm\sub\in\den{\type\to\typetwo}$ iff
     for all $\tmtwo\in\den\type$, $\tm\sub\subst\var\tmtwo\in\den\typetwo$.
     By taking $\sub'\defeq\sub\cup\subst\var\tmtwo$, this follows
     from the \ih{} and from the hypothesis on $\sub$.
    \item[($\to_E$)]
     We need to prove that for all $\sub\in\den{\Gamma}$,
      $(\tm\,\tmtwo)\sub\in\den{\typetwo}$.
      Note that $(\tm\,\tmtwo)\sub = (\tm\sub)(\tmtwo\sub)$.
      By \ih{} $\tm\sub\in\den{\type\to\typetwo}$, and therefore by \refdefp{valuation}{to}, either:
      \begin{itemize}
        \item $\tm\sub\toenf \ntm \in \SNe$: then $(\tm\sub)(\tmtwo\sub) \toes \ntm(\tmtwo\sub)$ since $\tmtwo\sub\in\SN$
        (by \ih{} and \reflemmap{candidates}{sn})
        Note that
        $\ntm(\tmtwo\sub)$ is neutral, and
        we conclude by \reflemmap{candidates}{vars} and \reflemmap{candidates}{closure}.
        % By \ih{} and \reflemmap{candidates}{sn} 
        % $\tmtwo\sub\in\SN$; therefore $\Ctx'\angled{\EFQ\tmthree}$ in inert and we conclude by \reflemmap{candidates}{vars}.
        % \item $\tm\sub\toenf \elctx\angled\tmthree \in \SNe$: then $(\tm\sub)(\tmtwo\sub) \toes \elctx\angled\tmthree\,(\tmtwo\sub)$
        %  and we conclude in a similar way as above.
        \item $\tm\sub\toenf \lambda\vartwo.\,\tmthree$: then $(\tm\sub)(\tmtwo\sub) \toes (\lambda\vartwo.\,\tmthree)\,(\tmtwo\sub) \toe \tmthree\subst\vartwo{\tmtwo\sub} \in \den\typetwo$
         by \refdefp{valuation}{to}. Conclude by \reflemmap{candidates}{closure}.
      \end{itemize}
    \item[($\bot_I$)]
     By the hypothesis, for every $\sub\in\den\Gamma$, $\tm\sub\in\den\bot$.
     We need to prove that $(\EFQ\tm)\sub\in\den\type$.
     By \reflemmap{candidates}{sn} $\tm\sub\in\SN$,
     and since $(\EFQ\tm)\sub = \EFQ{(\tm\sub)}$,
     $(\EFQ\tm)\sub$ is a neutral term.
     Conclude by \reflemmap{candidates}{vars}.
    \item[($\wedge_I$)] Let $\Gamma\vdash\tm_1\colon\type_1$ and 
     $\Gamma\vdash\tm_2\colon\type_2$:
     we need to prove that for every $\sub\in\den\Gamma$, $\Pair{\tm_1}{\tm_2}\sub\in\den{\type_1\land\type_2}$.
     Since $\Pair\tm\tmtwo\sub = \Pair{\tm\sub}{\tmtwo\sub}$, the claim follows
     from \refdefp{valuation}{and} and the \ih{}
     $\tm\sub\in\den\type$ and $\tmtwo\sub\in\den\typetwo$.
    \item[($\wedge_E$)] Let $\Gamma\vdash\tmtwo\colon\type_1\land\type_2$,
    and by \ih{} $\tmtwo\sub\in\den{\type_1\land\type_2}$ for every $\sub\in\den\Gamma$.
     We need to prove that $(\Proj 1 \tmtwo)\sub\in\den{\type_1}$ and $(\Proj 2 \tmtwo)\sub\in\den{\type_2}$
     for every $\sub\in\den\Gamma$. There are two cases:
     \begin{itemize}
       \item $\tmtwo\sub \toenf \ntm \in \SNe$:
        then $\Proj i {(\tmtwo\sub)} \toenf \Proj i \ntm$ and
        we conclude by \reflemmap{candidates}{vars} and \reflemmap{candidates}{closure} because that term is neutral.
       % \item $\tmtwo\sub \toenf \ntm \neq \EFQ\cdot\in \SNe$: then also $\Proj i (\tmthree\sub) \toe \Proj i \ntm \in \SNe$,
       % and we conclude by \reflemmap{candidates}{vars}.
       \item $\tmthree\sub \toenf \Pair{\tm_1}{\tm_2}$ for some $\tm_1\in\den{\type_1}$ and $\tm_2\in\den{\type_2}$:
        therefore $\Proj i (\tmthree\sub) \toes \Proj i {\Pair{\tm_1}{\tm_2}}$ \\$\toe \tm_i\in\den{\type_i}$.
        Conclude by \reflemmap{candidates}{closure} since $(\Proj i \tmthree)\sub = \Proj i (\tmthree\sub)$.
     \end{itemize}
    \item[($\vee_I$)]
      We discuss the case of $\In 1{}$; the case of $\In 2 {}$ is symmetric.
      Let $\Gamma\vdash\tm\colon\type$, and by \ih{} $\tm\sub\in\den\type$ for every $\sub\in\den\Gamma$.
      We show that also $(\In 1 \tm)\sub \in\den{\type\lor\typetwo}$ for every $\sub\in\den\Gamma$.
      Note that $(\In 1 \tm)\sub = \In 1 {(\tm\sub)}$, and conclude by
      \ih{} and \refdefp{valuation}{or}.
    \item[($\vee_E$)] This case is just a simplified version of the following argument for the Harrop rule.
    \item[(Harrop)]
    We need to prove that 
     $(\harrop {\Bind\var\tm} {\Bind\vartwo\tmtwo_1} {\Bind\vartwo\tmtwo_2})\sub \in\den\typefour$
     for every $\sub\in\den\Gamma$.
     We first note that $(\harrop {\Bind\var\tm} {\Bind\vartwo\tmtwo_1} {\Bind\vartwo\tmtwo_2})\sub = %
     \harrop {\Bind\var\tm\sub} {\Bind\vartwo\tmtwo_1\sub} {\Bind\vartwo\tmtwo_2\sub}$
     (assuming by renaming that $\var$ and $\vartwo$ do not occur in $\sub$).
     Let $\sub'\defeq\sub\cup\subst\var\var$. $\tm\sub=\tm\sub'$ and by \ih{}
     $\tm\sub' \in \den{\type_1\lor\type_2}$.
     % By \refdef{closure} and \refdefp{valuation}{or} we proceed by cases:
     There are three cases:
     \begin{itemize}
       \item $\tm\sub' \toenf \In i \tmthree_i$ for $\tmthree_i\in\den{\type_i}$:
        then also $\harrop {\Bind\var\tm\sub} {\Bind\vartwo\tmtwo_1\sub} {\Bind\vartwo\tmtwo_2\sub}
         \toes \tmtwo_i\sub\subst\vartwo{\lambda\var.\,\tmthree_i}$.
         In order to be able to use the \ih{} we need to show that $\sub\cup\subst\vartwo{\lambda\var.\,\tmthree_i}\in\den{\Gamma,\vartwo\colon\neg\typetwo\to\type_i}$,
         \ie{} that $\lambda\var.\,\tmthree_i \in \den{\neg\typetwo\to\type_i}$,
         that by definition holds iff for every $\tm'\in\den{\neg\typetwo}$, $\tmthree_i\subst\var{\tm'}\in\den{\type_i}$.
         In order to show the latter, take $\sub''\defeq\sub\cup\subst\var{\tm'}$: then by \ih{} $\tm\sub''\in\den{\type_1\lor\type_2}\subseteq\SN$,
         and therefore by \reflemma{toes-subst} $\tm\sigma'' = \tm\sigma'\subst\var{\tm'}\toes \In i {(\tmthree_i\subst\var{\tm'})} \toenf \In i {\tmthree_i'}$ for $\tmthree_i\subst\var{\tm'} \toenf \tmthree_i'$.
         By \refdefp{valuation}{or}
          $\tmthree_i'\in\den{\type_i}$, but also $\tmthree_i\subst\var{\tm'}$ by \reflemmap{candidates}{closure}, and we conclude.
         \item $\tm\sub'\toenf \Ctx\angled{\EFQ\tmthree} \in \SNe$: then also $\harrop {\Bind\var\tm\sub} {\Bind\vartwo\tmtwo_1\sub} {\Bind\vartwo\tmtwo_2\sub}
         \toes \tmtwo_1\sub\subst\vartwo{\lambda\var.\,\EFQ \tmthree}$.
         As above, in order to use the \ih{} and conclude we only need to prove that
         $(\lambda\var.\,\EFQ \tmthree)\in\den{\neg\typetwo\to\type_1}$. By \refdefp{valuation}{to}, this is the case
         if and only if for all $\tmthree'\in\den{\neg\typetwo}$, \\$(\EFQ \tmthree)\subst\var{\tmthree'} \in\den{\type_1}$.
         This is proved similarly as the point above, and it follows by\\ \reflemma{toes-subst} and the definition of inert terms.
       \item $\tm\sub'\toenf \elctx\angled\varthree\in\SNe$:
       we conclude as usual because $\harrop {\Bind\var\elctx\angled\varthree} {\Bind\vartwo\tmtwo_1\sub} {\Bind\vartwo\tmtwo_2\sub}$ is neutral as well.
     \end{itemize}
  \end{itemize}
\end{proof}

\CopyTheorem{Strong normalization of \KP{}}{kp-sn}%
  {If $\Gamma\vdash_{\text{\KP{}}}\tm\colon\type$, then $\tm$ is strongly normalizing.}%
  % copiare con \PasteTheorem{kp-sn}
\begin{proof}
  By \reflemma{aux-sn}, $\tm\sub\in\den\type$ for every $\sub\in\den\Gamma$.
  We now take $\sub$ as the \emph{identity substitution}, mapping the variables
  in $\Gamma$ to themselves. Note that this is an allowed substitution since
  variables are neutral terms and therefore are contained in the denotation of
  every proposition (\reflemmap{candidates}{vars}).
  It follows that $\tm=\tm\sub\in\den\type$, and we conclude because $\den\type$ contains only
  $\SN$ terms by \reflemmap{candidates}{sn}.
\end{proof}